\documentclass{ap-jnmp}

\def\DIS{\displaystyle}

\def\qed{\hfill\hbox{$\Box$}\vspace{10pt}\break}
\def\hugesymbol#1{\mbox{\strut\rlap{\smash{\Huge$#1$}}\quad}}
\def\C{{\mathbb C}}
\def\Z{{\mathbb Z}}
\def\Q{{\mathbb Q}}
\def\F{{\mathbb F}}
\def\R{{\mathbb R}}
\def\P{{\mathbb P}}
\def\A{{\mathbb A}}
\def\II{${}_{\mbox{\scriptsize{II} }}$}
\def\sp{\mbox{\scriptsize p}}

\copyrightauthor{}

\title{The space of initial conditions and the property of an almost good reduction in discrete Painlev\'{e} II equations over finite fields}

\author{\footnotesize Masataka Kanki}

\address{Graduate School of Mathematical Sciences, University of Tokyo,\\
3-8-1 Komaba, Tokyo 153-8914, Japan\\
\email{kanki@ms.u-tokyo.ac.jp}}

\author{Jun Mada}

\address{College of Industrial Technology, \\
Nihon University, 2-11-1 Shin-ei, Narashino, Chiba 275-8576, Japan\\\email{mada.jun@nihon-u.ac.jp}}

\author{\footnotesize Tetsuji Tokihiro}

\address{Graduate School of Mathematical Sciences, University of Tokyo,\\
3-8-1 Komaba, Tokyo 153-8914, Japan\\
\email{toki@ms.u-tokyo.ac.jp}}

\begin{document}

\maketitle
\thispagestyle{empty}

\vphantom{\vbox{%
\begin{history}
\received{(Day Month Year)}
\revised{(Day Month Year)}
\accepted{(Day Month Year)}
\end{history}
}}

\begin{abstract}
We investigate the discrete Painlev\'{e} equations (dP\II and $q$P\II) over finite fields. We first show that they are well defined by extending the domain according to the theory of the space of initial conditions.
Then we treat them over local fields and observe that they have a property that is similar to the good reduction of dynamical systems over finite fields.
We can use this property, which can be interpreted as an arithmetic analogue of singularity confinement, to avoid the indeterminacy of the equations over finite fields and to obtain special solutions from those defined originally over fields of characteristic zero.
\end{abstract}

\keywords{discrete Painlev\'{e} equation; finite field; good reduction; space of initial condition.}
\ccode{2000 Mathematics Subject Classification: 37K10, 34M55, 37P25}

\section{Introduction}
\label{sec1}

In this article, we study the discrete Painlev\'{e} equations over finite fields.
The discrete Painlev\'{e} equations are non-autonomous, integrable mappings which tend to some continuous Painlev\'{e} equations for appropriate choices of the continuous limit \cite{RGH}.
When we treate a discrete Painlev\'{e} equation over a finite field, we encounter the problem that its time evolution is not always well defined. 
This problem cannot be solved even if we extend the domain from $\F_q$ to the projective space $\P\F_q$, because $\P\F_q$ is no longer a field;
we cannot determine the values such as $\frac{0}{0},\ 0 \cdot \infty,\ \infty+\infty$ and so on.

There may be two strategies to define the time evolution over a finite field without inconsistencies.
One is to reduce the domain so that the time evolution will not pass the indeterminate states, and the other is to extend the domain so that it can include all the orbits. 
We take the latter strategy and adopt two approaches.

The first approach is the application of the theory of space of initial conditions developed by Okamoto \cite{Okamoto} and Sakai \cite{Sakai}.
We show that the dynamics of the equations over finite fields can be well defined in the space of initial conditions. 

The second approach we adopt is closely related to the theory of arithmetic dynamical systems, which concerns the dynamics over arithmetic sets such as $\mathbb{Z}$ or $\mathbb{Q}$ or a number field that is of number theoretic interest \cite{Silverman}.
In arithmetic dynamics, the change of dynamical properties of polynomial or rational mappings give significant information when reducing them modulo prime numbers. 
The mapping is said to have good reduction if, roughly speaking, the reduction commutes with the mapping itself \cite{Silverman}.
Linear fractional transformations in $\mbox{PGL}_2$ are typical examples of mappings with a good reduction.
Recently bi-rational mappings over finite fields have been investigated in terms of integrability
\cite{Roberts}.
Since all the orbits are cyclic as far as the mapping is closed over the (projective) space of a finite field and the integrable mapping has a conserved quantity, one can 
estimate the distribution of orbit length by using Hesse-Weil bounds or numerical calculations.
The QRT mappings \cite{QRT} over finite fields have been studied in detail by choosing the parameter values so that indeterminate points are avoided \cite{Roberts2}.
They have a good reduction over finite fields.

We prove that, although most of the integrable mappings with singularities do not have a good reduction modulo a prime in general,
they have an \textit{almost good reduction}, which is a generalization of good reduction.
We apply the method to the $q$-discrete Painlev\'{e} II equation ($q$P\II ).
The time evolution of the discrete  Painlev\'{e} equations can be well defined generically, even when not defined over the projective space of the field.
In particular, the reduction from a local field $\Q_p$ to a finite field $\F_p\cup\{\infty\}$ is shown to be well defined and is used to obtain some special solutions directly from those over fields of characteristic zero such as $\Q$ or $\R$.
%
%
\section{The dP\II equation and its space of initial conditions}
%
%
A discrete Painlev\'{e} equation is a non-autonomous and nonlinear second order ordinary difference equation with several parameters.
When it is defined over a finite field, the dependent variable takes only a finite number of values and its time evolution will attain an indeterminate state in many cases for generic values of the parameters and initial conditions.
For example, the dP\II equation is defined as 
\begin{equation}
u_{n+1}+u_{n-1}=\frac{z_n u_n+a}{1-u_n^2}\quad (n \in \mathbb{Z}),
\label{dP2equation}
\end{equation}
where $z_n=\delta n + z_0$ and $a, \delta, z_0$ are constant parameters \cite{NP}.
Let $q=p^k$ for a prime $p$ and a positive integer $k \in \Z_+$.
When \eqref{dP2equation} is defined over a finite field $\F_{q}$,
the dependent variable $u_n$ will eventually take values $\pm 1$ for generic parameters and initial values $(u_0,u_1) \in \F_{q}^2$, 
and we cannot proceed to evolve it.
If we extend the domain from $\F_{q}^2$ to $(\P\F_q)^2=(\F_q\cup\{\infty\})^2$, $\P\F_q$ is not a field and we cannot define arithmetic operation in \eqref{dP2equation}. 
To determine its time evolution consistently, we have two choices:
One is to restrict the parameters and the initial values to a smaller domain so that the singularities do not appear.
The other is to extend the domain on which the equation is defined.
In this article, we will adopt the latter approach.
It is convenient to rewrite \eqref{dP2equation} as:
\begin{equation}
\left\{
\begin{array}{cl}
x_{n+1}&=\dfrac{\alpha_n}{1-x_n}+\dfrac{\beta_n}{1+x_n}-y_{n},\\
y_{n+1}&=x_n,
\end{array}
\right.
\label{dP2}
\end{equation}
where $\alpha_n:=\frac{1}{2}(z_n+a),\ \beta_n:=\frac{1}{2}(-z_n+a)$.
Then we can regard \eqref{dP2} as a mapping defined on the domain $\F_q \times \F_q$.
To resolve the indeterminacy at $x_n = \pm 1$, we apply the theory of the state of initial conditions developed by Sakai \cite{Sakai}.
First we extend the domain to $\P\F_q \times \P\F_q$, and then blow it up at four points $(x,y)=(\pm 1, \infty), (\infty, \pm 1)$ 
to obtain the space of initial conditions:
\begin{equation}
\tilde{\Omega}^{(n)}:=\mathcal{A}_{(1,\infty)}^{(n)}\cup \mathcal{A}_{(-1,\infty)}^{(n)}\cup \mathcal{A}_{(\infty,1)}^{(n)}\cup \mathcal{A}_{(\infty,-1)}^{(n)},
\label{omega}
\end{equation}
where $\mathcal{A}_{(1,\infty)}^{(n)}$ is the space obtained from the two dimensional affine space $\A^2$ by blowing up twice as
\begin{align*}
\mathcal{A}_{(1,\infty)}^{(n)}&:=\left\{ \left((x-1,y^{-1}),[\xi_1:\eta_1],[u_1:v_1]  \right)\ \Big|\ \right. \\
&\qquad \eta_1 (x-1)=\xi_1 y^{-1},
(\xi_1+\alpha_n \eta_1)v_1=\eta_1(1-x)u_1 \ \Big\} \; \subset \A^2 \times \P \times \P.  
\end{align*}
Similarly, 
\begin{align*}
\mathcal{A}_{(-1,\infty)}^{(n)}&:=\left\{ \left((x+1,y^{-1}),[\xi_2:\eta_2],[u_2:v_2]  \right)\ \Big|\ 
\right.\\
&\qquad \qquad \eta_2 (x+1)=\xi_2 y^{-1},(-\xi_2+\beta_n \eta_2)v_2=\eta_2(1+x)u_2 \ \Big\},\\
\mathcal{A}_{(\infty,1)}^{(n)}&:=\left\{ \left((x^{-1},y-1),[\xi_3:\eta_3],[u_3:v_3]  \right)\ \Big|\ 
\right. \\
& \qquad \qquad \xi_3 (y-1)=\eta_3 x^{-1}, (\eta_3+\alpha_n \xi_3)v_3=\xi_3(1-y)u_3 \ \Big\},\\
\mathcal{A}_{(\infty,-1)}^{(n)}&:=\left\{ \left((x^{-1},y+1),[\xi_4:\eta_4],[u_4:v_4]  \right)\ \Big|\ 
\right. \\
& \qquad \qquad \xi_4 (y+1)=\eta_4 x^{-1}, (-\eta_4+\beta_n \xi_4)v_3=\xi_4(1+y)u_4 \ \Big\}.
\end{align*}  
The bi-rational map \eqref{dP2} is extended to the bijection $\tilde{\phi}_n: \ \tilde{\Omega}^{(n)} \rightarrow \tilde{\Omega}^{(n+1)}$ 
which decomposes as $\tilde{\phi}_n:=\iota_n \circ \tilde{\omega}_n$. 
Here $\iota_n$ is a natural isomorphism which gives $\tilde{\Omega}^{(n)} \cong  \tilde{\Omega}^{(n+1)}$, that is,
on $\mathcal{A}_{(1,\infty)}^{(n)}$ for instance, $\iota_n$ is expressed as 
\begin{align*}
&\left((x-1,y^{-1}),[\xi :\eta ],[u :v ]  \right) \in  \mathcal{A}_{(1,\infty)}^{(n)} \\
&\rightarrow \quad
\left((x-1,y^{-1}),[\xi -\delta/2\cdot\eta:\eta ],[u :v ]  \right) \in  \mathcal{A}_{(1,\infty)}^{(n+1)}.
\end{align*}

The automorphism $\tilde{\omega}_n$ on $\tilde{\Omega}^{(n)}$ is induced from \eqref{dP2} and gives the mapping
\[
\mathcal{A}_{(1, \infty)}^{(n)} \rightarrow \mathcal{A}_{(\infty,1)}^{(n)}, \;
\mathcal{A}_{(\infty,1)}^{(n)} \rightarrow \mathcal{A}_{(-1,\infty)}^{(n)}, \;
\mathcal{A}_{(-1, \infty)}^{(n)} \rightarrow \mathcal{A}_{(\infty,-1)}^{(n)}, \;
\mathcal{A}_{(\infty,-1)}^{(n)} \rightarrow \mathcal{A}_{(1,\infty)}^{(n)}.
\]
Under the map $\mathcal{A}_{(1, \infty)}^{(n)} \rightarrow \mathcal{A}_{(\infty,1)}^{(n)}$,
\begin{align*}
x=1 \ \rightarrow \ E_2^{(\infty,1)} &\qquad u_3=\left(y-\frac{\beta_n}{2}\right)v_3, \\
E_1^{(1,\infty)} \ \rightarrow \ E_1^{(\infty,1)} &\qquad [\xi_1:-\eta_1]=[\alpha_n \xi_3+\eta_3:\xi_3], \\
E_2^{(1,\infty)} \ \rightarrow \ y'=1 &\qquad x'=\frac{u_1}{v_1}+\frac{\beta_n}{2},
\end{align*}
where $(x,y) \in \mathcal{A}_{(1, \infty)}^{(n)}$, $(x',y')\in \mathcal{A}_{(\infty,1)}^{(n)}$, $E_1^{\sp}$ and $E_2^{\sp}$ are the exceptional curves in $\mathcal{A}_{\sp}^{(n)}$ obtained by the first blowing up and the second blowing up respectively at the point p $\in \{(\pm 1, \infty),(\infty,\pm 1)  \}$. 
Similarly under the map $\mathcal{A}_{(\infty,1)}^{(n)} \rightarrow \mathcal{A}_{(-1,\infty)}^{(n)}$,
\begin{align*}
E_1^{(\infty,1)} \ \rightarrow \ E_1^{(-1,\infty)} &\qquad [\xi_3:\eta_3]=[\eta_2:(\beta_n-\alpha_n) \eta_2-\xi_2], \\
E_2^{(\infty,1)} \ \rightarrow \ E_2^{(-1,\infty)} &\qquad [u_3:v_3]=[-\beta_n u_2: \alpha_n v_2].
\end{align*}
The mapping on the other points are defined in a similar manner.
Note that $\tilde{\omega}_n$ is well-defined in the case $\alpha_n=0$ or $\beta_n=0$.
In fact, for $\alpha_n=0$, $E_2^{(1,\infty)}$ and $E_2^{(\infty,1)}$ can be identified with the lines $x=1$ and $y=1$ respectively. 
Therefore we have found that, through the construction of the space of initial conditions, the dP\II equation can be well-defined over finite fields.
However there are some unnecessary elements in the space of initial conditions when we consider a finite field, because we are working on a discrete topology and do not need continuity of the map. 
Let $\tilde{\Omega}^{(n)}$ be the space of initial conditions and $|\tilde{\Omega}^{(n)}|$ be the number of elements of it.
For the dP\II equation, we obtain $|\tilde{\Omega}^{(n)}|=(q+1)^2-4+4(q+1)-4+4(q+1)=q^2+10q+1$, since $\P\F_q$ contains $q+1$ elements.
However an exceptional curve $E_1^{\sp}$ is transferred to another exceptional curve $E_1^{\sp'}$, and $[1:0] \in E_2^{\sp}$ to 
$[1:0] \in E_2^{\sp'}$ or to a point in $E_1^{\sp'}$. Hence we can reduce the space of initial conditions $\tilde{\Omega}^{(n)}$ to the minimal space of initial conditions $\Omega^{(n)}$ which is the minimal subset of $\tilde{\Omega}^{(n)}$ including $\P\F_q\times \P\F_q$, closed under the time evolution.
By subtracting unnecessary elements we find $|\Omega^{(n)}|=(q+1)^2-4+4(q+1)-4=q^2+6q-3$.
In summary, we obtain the following proposition:
\begin{theorem}
The domain of the dP\II equation over $\F_q$ can be extended to the minimal domain $\Omega^{(n)}$ on which the time evolution at time step $n$ is well defined. Moreover $|\Omega^{(n)}|=q^2+6q-3$. 
\end{theorem}

\begin{figure}
\centering
\includegraphics[width=11cm,bb=-152 152 510 662]{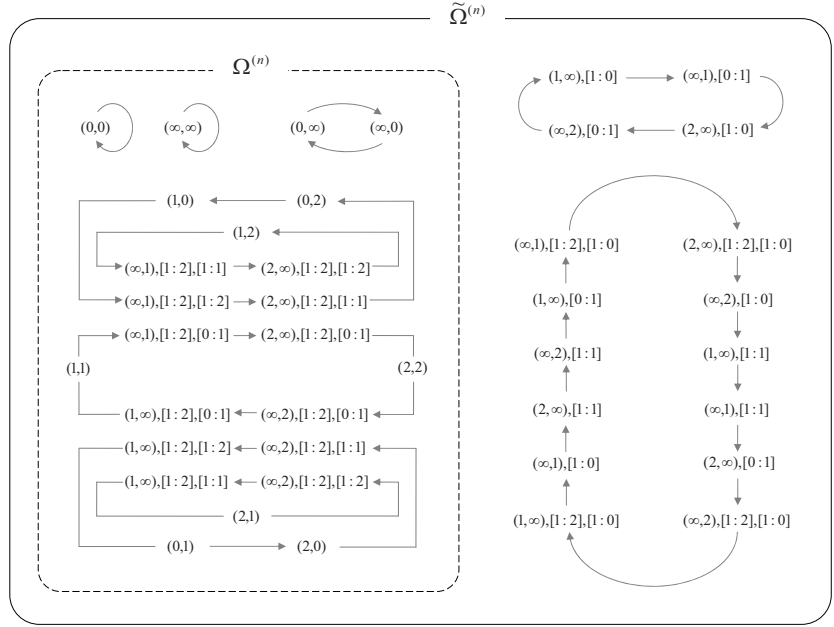}
\caption{The orbit decomposition of the space of initial conditions $\tilde{\Omega}^{(n)}$ and the reduced one $\Omega^{(n)}$ for $q=3$.}
\label{figure1}
\end{figure}
In figure \ref{figure1}, we show a schematic diagram of the map $\tilde{\omega}_n$ on $\tilde{\Omega}^{(n)}$, and its restriction map $\omega_n:=\tilde{\omega}_n|_{\Omega^{(n)}}$ on $\Omega^{(n)}$
with $q=3$, $\alpha_0=1$ and $\beta_0=2$.
We can also say that the figure \ref{figure1} is a diagram for the autonomous version of the equation \eqref{dP2} when $\delta=0$.
In the case of $q=3$, we have $|\tilde{\Omega}^{(n)}|=40$ and $|\Omega^{(n)}|=24$. 

The above approach is equally valid for the other discrete Painlev\'{e} equations and we can define them over finite fields by constructing isomorphisms on the spaces of initial conditions.
Thus we conclude that a discrete Painlev\'{e} equation can be well defined over a finite field by redefining the initial domain properly. 
However, for a general nonlinear equation, explicit construction of the space of initial conditions over a finite field is not so straightforward \cite{Takenawa} and it will not help us to obtain the explicit solutions.
In the next section, we show another extention of the space of initial conditions: we extend it to $\Z_p\times \Z_p$.  
%
%
\section{The $q$P\II equation over a local field and its reduction modulo a prime}
%
%
Let $p$ be a prime number and for each $x \in \Q$ ($x \ne 0$) write $x=p^{v_p(x)} \dfrac{u}{v}$ where $v_p(x), u, v \in \Z$ and $u$ and $v$ are coprime integers neither of which is divisible by $p$.
The $p$-adic norm $|x|_p$ is defined as $|x|_p=p^{-v_p(x)}$. ($|0|_p=0$.)
The local field $\Q_p$ is a completion of $\Q$ with respect to the $p$-adic norm. 
It is called the field of $p$-adic numbers and its subring $\Z_p:=\{x\in \Q_p | \ |x|_p \le 1\}$ is called the ring of $p$-adic integers \cite{Murty}. 
The $p$-adic norm satisfies a non-archimedean (ultrametric) triangle inequality 
$|x+y|_p \le \max[|x|_p,|y|_p ]$.
Let $\mathfrak{p}=p\Z_p=\left\{x \in \Z_p |\ v_p(x) \ge 1 \right\}$ be the maximal ideal of $\Z_p$.
We define the reduction of $x$ modulo $\mathfrak{p}$ as $\tilde{x}$: $\Z_p \ni x \mapsto \tilde{x} \in \Z_p/\mathfrak{p} \cong \F_p$.
Note that the reduction is a ring homomorphism.
The reduction is generalized to $\Q_p^{\times}$:
\[
\Q_p^{\times}\ni x=p^k u\ (u\in\Z_p^{\times})\mapsto
\left\{
\begin{array}{cl}
0 & (k>0)\\
\infty & (k<0)\\
\tilde{u} & (k=0)
\end{array}
\right. \in \P\F_p,
\]
which is no longer homomorphic.
For a rational map of the plane
\[
\phi(x,y)=(\phi_1(x,y), \phi_2(x,y))\in(\mathbb{Z}_p(x,y))^2:\ \mathcal{D} \subseteq \Z_p^2 \to \Z_p^2,
\]
defined on some domain $\mathcal{D}$,
\[
\tilde{\phi}(x,y)=(\tilde{\phi}_1(x,y), \tilde{\phi}_2(x,y))\in(\mathbb{F}_p(x,y))^2
\]
is defined as the map whose coefficients are all reduced.
The rational system $\phi$ is said to have a \textit{good reduction} (modulo $\mathfrak{p}$ on the domain $\mathcal{D}$) if we have $\widetilde{\phi(x,y)}=\tilde{\phi}(\tilde{x},\tilde{y})$ for any $(x,y) \in \mathcal{D}$ \cite{Silverman}.
We have defined a generalized notion in our previous letter and have explained its usefulness;
\begin{definition}[\cite{KMTT}]
A (non-autonomous) rational system $\phi_n$: $\Q_p^2 \to \Q_p^2$ $(n \in \Z)$ is said to have an almost good reduction modulo $\mathfrak{p}$ on the domain $\mathcal{D} \subseteq \Z_p^2$, if there
exists a positive integer $m_{\mbox{\rm \scriptsize p};n}$ for any $\mbox{\rm p}=(x,y) \in \mathcal{D}$ and time step $n$ such that
\begin{equation}
\widetilde{\phi_n^{m_{\mbox{\rm \tiny p};n}}(x,y)}=\widetilde{\phi_n^{m_{\mbox{\rm \tiny p};n}}}(\tilde{x},\tilde{y}),
\label{AGR}
\end{equation}
where $\phi_n^m :=\phi_{n+m-1} \circ \phi_{n+m-2} \circ \cdots \circ \phi_n$.
\end{definition} 
If we can take $m_{\mbox{\rm \scriptsize p};n}=1$, then the mapping has a good reduction.
Let us first review some of the previous findings in order to see the significance of the notion of \textit{almost good reduction}. Let us consider the mapping $\Psi_\gamma$:
\begin{equation}
\left\{
\begin{array}{cl}
x_{n+1}&=\dfrac{x_n+1}{x_n^\gamma y_n}\\
y_{n+1}&=x_n
\end{array}
\right.,
\label{discretemap}
\end{equation} 
where $\gamma \in \Z_{\ge 0}$ is a non-negative integer parameter. 
The map \eqref{discretemap} is known to be integrable if and only if $\gamma=0,1,2$.
When $\gamma=0,1,2$, the map \eqref{discretemap} belongs to the QRT family and is integrable in the sense that it has a conserved quantity. We also note that the system \eqref{discretemap} can be seen as an autonomous version of the $q$-discrete Painlev\'{e} I equation for $\gamma=0,1,2$. We have proved that, in this example, the integrability is equivalent to having almost good reduction.
\begin{proposition}[\cite{KMTT}]
The rational mapping \eqref{discretemap} has an almost good reduction modulo $\mathfrak{p}$ on the domain $\mathcal{D}=\{(x,y) \in \Z_p^2 \ |x \ne 0, y \ne 0\}$, if and only if $\gamma=0,1,2$.
\label{PropQRT}
\end{proposition}
In the previous work, the dP\II equation, too,  has been proved to have almost good reduction.
It is reasonable to postulate that almost good reduction is closely related to the integrability of maps of the plane.
We further clarify this point by demonstrating that $q$-discrete analogue of the Painlev\'{e} II equation also has almost good reduction. 
The $q$-discrete analogue of Painlev\'{e} II equation ($q$P\II equation) is a following $q$-difference equation:
\begin{equation}
(z(q\tau)z(\tau)+1)(z(\tau)z(q^{-1}\tau)+1)=\frac{a \tau^2 z(\tau)}{\tau-z(\tau)},
\label{qP2eq}
\end{equation}
where $a$ and $q$ are parameters \cite{Kajiwaraetal}.
It is also convenient to rewrite \eqref{qP2eq} as a system form
\begin{equation}
\Phi_n: \left\{
\begin{array}{cl}
x_{n+1}&=\dfrac{a(q^n\tau_0)^2x_n-(q^n\tau_0-x_n)(1+x_ny_n)}{x_n(q^n\tau_0-x_n)(x_ny_n+1)},\\
y_{n+1}&=x_n,
\end{array}
\right.
\label{qP2}
\end{equation}
where $\tau=q^n\tau_0$.
We can prove the following proposition:
\begin{proposition}
Suppose that $a, q, \tau_0$ are integers not divisible by $p$, then the mapping \eqref{qP2} has an almost good reduction   
modulo $\mathfrak{p}$ on the domain 
$\mathcal{D}:=\{(x,y)\in \Z_p^2\ |x \ne 0, x \ne q^n\tau_0\ (n\in\Z), xy+1 \ne 0\}$.
\label{PropqP2}
\end{proposition}
\begin{proof}
Let $(x_{n+1},y_{n+1})=\Phi_n(x_n,y_n)$. 
We only examine the cases $\tilde{x}_n=0$ or $\widetilde{q^n\tau_0}$ 
or $-\tilde{y}_n^{-1}$. (It is due to the fact that the mapping trivially has good reduction for other points $(x_n,y_n)$.)
We use the abbreviation $\tilde{q}=q, \tilde{\tau}=\tau, \tilde{a}=a$ for simplicity. 
 By direct computation, we obtain;\\
(i) If $\tilde{x}_n=0$ and $ -1+q^2-aq^2\tau^2+q^3\tau^2-q^2\tau \tilde{y}_n \ne 0$, 
\[
\widetilde{\Phi_n^3(x_n,y_n)} = \widetilde{\Phi_n^3}(\tilde{x}_n=0,\tilde{y}_n)=\left( \frac{ 1-q^2+aq^2\tau^2-q^3\tau^2-aq^4\tau^2+q^2\tau \tilde{y}_n}{q^2\tau( -1+q^2-aq^2\tau^2+q^3\tau^2-q^2\tau \tilde{y}_n ) }   , q^2\tau  \right).
\]
(ii) If $\tilde{x}_n=0$ and $ -1+q^2-aq^2\tau^2+q^3\tau^2-q^2\tau \tilde{y}_n = 0$, 
\[
\widetilde{\Phi_n^5(x_n,y_n)} = \widetilde{\Phi_n^5}(\tilde{x}_n=0,\tilde{y}_n) =\left( \frac{1-q^2+q^7\tau^2-aq^8\tau^2}{q^4\tau}, 0    \right).
\]
(iii) If $\tilde{x}_n=\tau$ and $1+\tau \tilde{y}_n\ne 0$,
\begin{align*}
&\widetilde{\Phi_n^3(x_n,y_n)} = \widetilde{\Phi_n^3}(\tilde{x}_n=\tau,\tilde{y}_n)\\
&\quad =\left(\frac{ 1-q^2+(a+q-aq^2)q^2\tau^2+(1-q^2)\tau\tilde{y}+(1-aq)q^3\tau^3\tilde{y}  }{q^2\tau(1+\tau\tilde{y}_n)}, 0 \right).
\end{align*}
(iv) If $\tilde{x}_n=\tau$ and $1+\tau \tilde{y}_n= 0$,
\[
\widetilde{\Phi_n^7(x_n,y_n)} = \widetilde{\Phi_n^7}(\tilde{x}_n=\tau,\tilde{y}_n)=\left(\frac{1}{aq^{12}\tau^3}, - aq^{12}\tau^3  \right).
\]
(v) If $\tilde{x}_n \tilde{y}_n+1=0$,
\[
\widetilde{\Phi_n^7(x_n,y_n)} = \widetilde{\Phi_n^7}(\tilde{x}_n=-\tilde{y}_n^{-1}, \tilde{y}_n)=\left(-\frac{1}{aq^{12}\tau^4\tilde{y}_n}, aq^{12}\tau^4\tilde{y}_n  \right).
\]
Thus we complete the proof. \end{proof}
From this proposition we can explicitly define the time evolution of the $q$P\II equation.
We can also consider the special solutions for qP\II equation \eqref{qP2eq} over $\P\F_p$. 
In \cite{HKW} it has been proved that \eqref{qP2eq} over $\C$ with $a=q^{2N+1}$ $(N \in \Z)$ is solved by
the functions given by
\begin{align}
z^{(N)} (\tau) &= 
\begin{cases}
\DIS \frac{g^{(N)} (\tau) g^{(N+1)} (q \tau)}{q^N g^{(N)} (q \tau) g^{(N+1)} (\tau)}
 & (N \ge 0) \\
\DIS \frac{g^{(N)} (\tau) g^{(N+1)} (q \tau)}{q^{N+1} g^{(N)} (q \tau) g^{(N+1)} (\tau)} & (N<0)
\end{cases}, \label{eq:gtoz} \\
g^{(N)} (\tau) &= 
\begin{cases}
\begin{vmatrix}
w(q^{-i+2j-1}\tau)
\end{vmatrix}_{1\le i,j\le N} & (N>0) \\
1 & (N=0) \\
\begin{vmatrix}
w(q^{i-2j} \tau)
\end{vmatrix}_{1\le i,j\le -N} & (N<0)
\end{cases}, \label{eq:det_sol_g}
\end{align}
where $w(\tau)$ is a solution of the $q$-discrete Airy equation:
\begin{equation}
w(q\tau)-\tau w(\tau)+w(q^{-1}\tau)=0. 
\label{dAiryeq}
\end{equation}
As in the case of the dP\II equation, we can obtain the corresponding solutions 
to \eqref{eq:gtoz} over $\P\F_p$ by reduction modulo $\mathfrak{p}$ according to the proposition \ref{PropqP2}.
For that purpose, we have only to solve \eqref{dAiryeq} over $\Q_p$.
By elementary computation we obtain:
\begin{equation}
w(q^{n+1}\tau_0)=c_1P_{n}(\tau_0;q)+c_0P_{n-1}(q\tau_0;q),
\label{Airy:sol}
\end{equation}
where $c_0,\ c_1$ are arbitrary constants and $P_n(x;q)$ is defined by the tridiagonal determinant: 
\[
P_n(x;q):=
\left|
\begin{array}{ccccc}
qx&-1&&&\\
-1&q^2x&-1&&\hugesymbol{0}\\
&\ddots&\ddots&\ddots& \\
&&-1&q^{n-1}x&-1\\
\hugesymbol{0}&&&-1&q^{n}x
\end{array}
\right|.
\]
The function $P_n(x;q)$ is the polynomial of $n$th order in $x$,
\[
P_n(x;q)=\sum_{k=0}^{[n/2]}(-1)^k a_{n;k}(q)x^{n-2k},
\]
where $a_{n;k}(q)$ are polynomials in $q$.
If we let $i \ll j $ denotes $i<j-1$, and 
$
c(j_1,j_2,...,j_k):=\sum_{r=1}^k (2j_r+1),
$
then, we have
\[
a_{n;k}=\sum_{1\le j_1 \ll j_2 \ll \cdots \ll j_k \le n-1} q^{n(n+1)/2 -c(j_1,j_2,...,j_k)}.
\] 
Therefore the solution of $q$P\II equation over $\P\F_p$ is obtained by reduction modulo $\mathfrak{p}$ from \eqref{eq:gtoz}, \eqref{eq:det_sol_g}
and \eqref{Airy:sol} over $\Q$ or $\Q_p$.

\section{Concluding remarks}
In this article we investigated the two types of the discrete Painlev\'{e} II equations and their reduction modulo a prime.
To avoid indeterminacy, we examined two approaches. 
One is to extend the domain by blowing up at indeterminate points. 
According to the theory of the space of initial conditions, this approach is possible for all the discrete Painlev\'{e} equations.
An interesting point is that the space of initial conditions over a finite field can be reduced to a smaller domain resulting from the discrete topology of the finite field.
The other is the reduction modulo a prime number from a local field, in particular, the field of $p$-adic numbers $\Q_p$.
We defined the notion of \textit{almost good reduction} which is an arithmetic analogue of passing the singularity confinement test, and proved that the $q$-discrete Painlev\'{e} II equation has this property.
Thanks to this property, not only the time evolution of the discrete Painlev\'{e} equations can be well defined, but also
a solution over $\Q$ (or $\Q_p$) can be directly transferred to a solution over $\P\F_p=\F_p\cup\{\infty\}$.
We presented the special solutions over $\P\F_p$.
We conjecture that this approach is equally valid in other discrete Painlev\'{e} equations.
We have recently proved almost good reduction property for several
$q$-discrete Painlev\'{e} equations \cite{Kankisigma}.
However, we have little results for other interesting equations,
such as chaotic maps, linearizable maps, and generalised Painlev\'{e} systems \cite{KNY}.
Furthermore, we expect that this `almost good reduction' criterion can be applied to finding higher order \textit{integrable} mappings in arithmetic dynamics, and
that a similar approach is also useful for the investigation of discrete partial difference equations such as soliton equations over finite fields \cite{DBK,KMT}. These problems are currently being investigated.     

\section*{Acknowledgments}
The authors wish to thank Professors K. M. Tamizhmani and  R. Willox for helpful discussions.
This work was partially supported by Grant-in-Aid for JSPS Fellows (24-1379).

\end{document}